\documentclass[11pt]{article}

\usepackage{fullpage}
\usepackage{amsmath,amsfonts,amsthm,mathrsfs,mathpazo,xspace,hyperref,graphicx}
\usepackage{endnotes}
\usepackage{color}
\usepackage{bm}
\usepackage{times}
\usepackage{amssymb,latexsym}

\newtheorem{theorem}{Theorem}

\newtheorem{lemma}[theorem]{Lemma}

\newtheorem{claim}[theorem]{Claim}

\newtheorem{corollary}[theorem]{Corollary}

\newenvironment{step}
  {
    \begin{enumerate}

  }
  {\end{enumerate}}

\newenvironment{algorithm*}[1]
  {
    \begin{center}
      \hrulefill\\
      \textbf{#1}
  }
  {
    \vspace{-1\baselineskip}
    \hrulefill
    \end{center}
  }

\newenvironment{protocol*}[1]
  {
    \begin{center}
      \hrulefill\\
      \textbf{#1}
  }
  {
    \vspace{-1\baselineskip}
    \hrulefill
    \end{center}
  }

\DeclareMathOperator{\poly}{poly}

\newcommand{\ket}[1]{|#1\rangle}
\newcommand{\bra}[1]{\langle#1|}
\newcommand{\Tr}{\mbox{\rm Tr}}
\newcommand{\SWAP}{\mbox{\rm SWAP}}
\newcommand{\DEC}{\mbox{\rm DEC}}
\newcommand{\CHECK}{\mbox{\rm CHECK}}
\newcommand{\CK}{\mbox{\rm CK}}

\newcommand{\C}{\ensuremath{\mathbb{C}}}
\newcommand{\N}{\ensuremath{\mathbb{N}}}

\newcommand{\mH}{\ensuremath{\mathcal{H}}}

\newcommand{\rQ}{\ensuremath{\mathrm{Q}}}
\newcommand{\rR}{\ensuremath{\mathrm{R}}}
\newcommand{\rP}{\ensuremath{\mathrm{P}}}
\newcommand{\rS}{\ensuremath{\mathrm{S}}}

\newcommand{\cC}{\ensuremath{\mathcal{C}}}
\newcommand{\cD}{\ensuremath{\mathcal{D}}}

\newcommand{\setft}[1]{\mathrm{#1}}
\newcommand{\density}[1]{\setft{D}\!\left(#1\right)}
\newcommand{\lin}[1]{\setft{L}\!\left(#1\right)}
\newcommand{\pos}[1]{\setft{Pos}\!\left(#1\right)}

\newcommand{\QMA}{\mathrm{QMA}}
\newcommand{\QCMA}{\mathrm{QCMA}}
\newcommand{\NP}{\mathrm{NP}}
\newcommand{\MIP}{\mathrm{MIP}}
\newcommand{\QMIP}{\mathrm{QMIP}}
\newcommand{\NEXP}{\mathrm{NEXP}}
\newcommand{\EXP}{\mathrm{EXP}}

\newcommand{\beq}{\begin{eqnarray}}
\newcommand{\eeq}{\end{eqnarray}}

\newcommand{\Id}{\ensuremath{\mathop{\rm Id}\nolimits}}

\newcommand{\eps}{\varepsilon}

\definecolor{mygrey}{gray}{0.50}

\begin{document}

\title{A multiprover interactive proof system for \\the local Hamiltonian problem}
\author{Joseph Fitzsimons\thanks{Singapore University of Technology and Design and Centre for Quantum Technologies, National University of Singapore, Singapore. Email: \texttt{joe.fitzsimons@nus.edu.sg}. } \qquad\qquad Thomas Vidick\thanks{California Institute of Technology, Pasadena, CA, USA. Email: \texttt{vidick@cms.caltech.edu}.}}
\date{}
\maketitle

\begin{abstract}
We give a quantum interactive proof system for the local Hamiltonian problem on $n$ qubits in which (i) the verifier has a single round of interaction with five entangled provers, (ii) the verifier sends a classical message on $O(\log n)$ bits to each prover, who reply with a constant number of qubits, and (iii) completeness and soundness are separated by an inverse polynomial in $n$. As the same class of proof systems, without entanglement between the provers, is included in $\QCMA$, our result provides the first indication that quantum multiprover interactive proof systems with entangled provers may be strictly more powerful than unentangled-prover interactive proof systems. A distinguishing feature of our protocol is that the completeness property requires honest provers to share a large entangled state, obtained as the encoding of the ground state of the local Hamiltonian via an error-correcting code. Our result can be interpreted as a first step towards a multiprover variant of the quantum PCP conjecture. 
\end{abstract}

\section{Introduction}

The PCP theorem~\cite{AroSaf98JACM,AroLunMotSudSze98JACM} asserts that any language in $\NP$ admits proofs of membership that can be efficiently verified using a randomized procedure which makes the correct decision with high probability while only ever reading a constant number of bits of the proof. An equivalent formulation of the PCP theorem, that has been particularly useful in applications to hardness of approximation~\cite{Hastad01} as well as in devising further improvements to the theorem~\cite{Raz98}, uses the language of multiplayer games. A two-player game $G$ is specified by question sets $Q,Q'$, answer sets $A,A'$, a distribution $\pi$ on $Q\times Q'$ and a verification criterion $V\subseteq (A\times A') \times (Q\times Q')$. The value $\omega$ of $G$ is defined as the maximum, over all assignments $f:Q\to A$, $f':Q'\to A'$, of the average number of valid answers given by the assignments: $\omega(G) = \sup_{f,f'}\sum_{q,q'} \pi(q,q') V(f(q),f'(q');q,q')$. The PCP theorem is equivalent to the statement that $\omega(G)$ is NP-hard to approximate to within a constant additive factor, even for the case of answer sets $A$, $A'$ of constant size. To see the connection, consider the following ``consistency game'': the verifier, instead of directly reading bits $i_1,\ldots,i_k$ of the proof, asks a first player for the entries at those locations and a second player for the entry corresponding to a single location $i_j$, where $j$ is chosen uniformly at random in $\{1,\ldots,k\}$. The verifier accepts if and only if the first player's answers correspond to entries that he would have accepted had he read them directly from the proof, \emph{and} the second player's answer is consistent with that of the first. It is not hard to see that the value of the consistency game is directly related to the fraction of checks satisfied by the optimal PCP proof, so that the respective complexities of deciding whether either is close to $1$ (under the appropriate gap promise) are identical. 

The quantum analogue of the local proof checking problem was introduced by Kitaev~\cite{KitSheVya02}. An instance of the $k$-local Hamiltonian problem (LH) is specified by $m$ local Hamiltonians $H_1,\ldots,H_m$, where each $H_i$ is a Hermitian matrix of norm at most $1$ acting on at most $k$ out of a total of $n$ qubits. The instance is positive if there exists a quantum proof (a quantum state $\ket{\Psi}$ on the $n$ qubits) satisfying a fraction at least $(1-a)$ of the constraints; precisely, if $H=\sum_i H_i$ (where each $H_i$ is implicitly tensored with the identity on the remaining qubits) has an eigenvalue at most $am$. If all eigenvalues of $H$ are larger than $bm$ for some $b> a$ the instance is negative. The introduction of the local Hamiltonian problem initiated what is now the burgeoning field of Hamiltonian complexity~\cite{Osborne12hc,GharibianHL14qhc}, expanding well beyond the initial formal connection with classical constraint satisfaction problems to encompass the computational study of a range of problems motivated by condensed-matter physics. 

Kitaev proved the ``quantum Cook-Levin theorem'': he introduced the class $\QMA$ of languages that admit efficiently verifiable quantum proofs, and showed that the local Hamiltonian problem is $\QMA$-complete for some $a,b$ satisfying $b-a = \Theta(\poly^{-1}(n))$. The natural question of whether a quantum analogue of the PCP theorem holds was first posed in~\cite{AharonovN02qma}; it asks whether the local Hamiltonian problem remains $\QMA$-hard for values $b-a=\Omega(1)$. This problem has captured the imagination of many researchers~\cite{Aaronson09qpcp,Has13qpcp,HasF13qpcp}, but very little is known. If anything recent results~\cite{BrandaoH13product,AharonovAV13qpcp} place strong limitations on the parameters, including the locality $k$ or the degree of the constraint graph, for which the conjecture may be valid, showing that it may only hold for ranges of parameters that appear to be much more limited than those for which the classical PCP theorem is known to be true. 

\medskip

In this paper we shed new light on the complexity of the local Hamiltonian problem by recasting it in the language of quantum interactive proofs with entangled provers. In doing so we are motivated by the existing classical connection between local proof verification and multiplayer games, which as already mentioned has been instrumental both in the development of the PCP theorem (and in particular its second proof by Dinur~\cite{Dinur07pcp}) and for applications. Does this connection extend to the quantum setting? While quantum multiprover interactive proof systems have been intensely studied for their own sake~\cite{KobMat03JCSS,KKMV09,IV12}, prior to our work no nontrivial relation was known between the class $\QMA_\EXP$, the exponentially scaled-up version of $\QMA$, and the classes $\QMIP^*$ or $\QMIP$ of languages having quantum interactive proof systems with entangled or unentangled provers respectively. In fact, the latter is known to equal $\NEXP$~\cite{KobMat03JCSS}, while the former was only recently shown to \emph{contain} $\NEXP$~\cite{IV12}. However, no upper bound on $\QMIP^*$ is known, so that one may ask --- could $\QMIP^*$ be a \emph{larger} class than $\QMIP=\NEXP$? The only distinction between the two classes is the presence of entanglement between the provers, which until now (and with some rare exceptions~\cite{KKMV09}) has for the most part been understood as a nefarious resource that could used by the provers in order to break a protocol's soundness. Giving a positive answer to the question, however, requires finding a \emph{beneficial} use of entanglement, as it entails devising a protocol in which even honest provers are \emph{required} to share an entangled state over a superpolynomial number of qubits in order to succeed on positive instances.\footnote{The class $\QMIP^{(l.e.)}$ of languages having quantum multiprover interactive proof systems in which the provers share an entangled state on at most a polynomial number of qubits is also known to be included in $\NEXP$\cite{KobMat03JCSS}.}

A natural target for going beyond $\NEXP\subseteq\QMIP^*$ consists in devising protocols establishing the inclusion of $\QMA_{\EXP}$ in $\QMIP^*$. Proving such inclusion, however, immediately runs into a number of serious difficulties. To see why, consider the following attempt at designing a quantum interactive proof system for the local Hamiltonian problem that mimics the classical construction of the consistency game (which, as described earlier, easily leads to a proof of $\NEXP\subseteq\MIP$ assuming the PCP theorem). Suppose that the first player is asked to provide a constant-sized subset of the proof qubits, corresponding to a local constraint $H_j$ which the verifier can then check. In the classical case, the second player is asked for just one of the bits asked to the first player; this is used to verify   that the first players' answers to any of the bits he was asked about depends on that bit only, and not on the subset of which it is part. In the quantum case this approach is all but ruled out by the no-cloning principle: any given proof qubit can be placed in the hands of one player only, but it cannot be duplicated! Hence the direct quantum analogue of the consistency game does not have \emph{completeness}: even satisfiable instances of the local Hamiltonian problem may not lead to a winning strategy for the players. 

Natural workarounds to this difficulty run into different obstacles. For instance, consider splitting the proof (e.g. the ground state of the local Hamiltonian instance) qubits into two (or more) sets $S_1$ and $S_2$, and only asking prover $i$ for qubits coming from set $S_i$. While this leads to a game which does have perfect completeness, the fact that the sets need to be specified a priori can, at least in some cases, prevent the \emph{soundness} property from holding. To see why, consider the simple example of a one-dimensional nearest-neighbor Hamiltonian in which each term is a projection on the orthogonal complement of an EPR pair split across two adjacent qubits. This Hamiltonian is highly frustrated, as any qubit can only form an EPR pair with its left \emph{or} right neighbor, not both. Nevertheless, the corresponding game in which $S_1$ (resp. $S_2$) is  the set of all even-numbered (resp. odd-numbered) qubits has a perfect strategy: the players share a single EPR pair and systematically send back their respective half, independently of the question they are asked! Although in this particular case the issue is easily fixed by choosing a different splitting of the proof qubits, in general it seems like any such splitting will be arbitrary and could be taken advantage of by the provers. 

\subsection{Results}

Our main result is the design of an interactive proof system for the local Hamiltonian problem which circumvents the aforementioned difficulties. This is the first time a multiprover interactive proof system is given for a $\QMA$-complete, instead of $\NP$-complete, problem, and it provides strong indication that entangled proof systems may be strictly more powerful than their unentangled counterparts. Formally, we show the following.  

\begin{theorem}\label{thm:main}
Let $k$ be an integer. There exists constants $C,c>0$ depending on $k$ only such that the following holds. Let $H=\sum_{i=1}^m H_i$ be an instance of the $k$-local Hamiltonian problem on $n$ qubits, such that the number of constraints is $m=\poly(n)$. There exists a one-round interactive protocol between a quantum polynomial-time verifier and $r=5$ entangled quantum provers such that: 
\begin{itemize}
\item The verifier sends $O(\log n)$-bit classical messages to each prover,
\item The provers respond with at most $k$ qubits each,
\item If there exists a state $\ket{\Gamma}$ such that $\bra{\Gamma}H\ket{\Gamma}\leq am$ then there is a strategy for the provers that is accepted with probability at least $1-a/2$, 
\item If for every state $\ket{\Psi}$, $\bra{\Psi}H\ket{\Psi}\geq bm$ then any strategy of the provers is accepted with probability at most $1-Cb/n^{c}$.
\end{itemize}
\end{theorem}

The local Hamiltonian problem is known to be $\QMA$-complete for $k=2$, $a$ that is exponentially small and $b$ at least an inverse polynomial~\cite{KempeKR06lh}. The following corollary, which we state using the language of multiplayer games, is thus a direct consequence of Theorem~\ref{thm:main}: 

\begin{corollary}\label{cor:qma}
The problem of approximating, to within an additive inverse polynomial, the referee's maximum acceptance probability in a quantum multiplayer game in which questions from the referee are classical on $O(\log n)$-bits and answers from the players are quantum on $O(1)$ qubits is $\QMA$-hard. Furthermore the same holds when restricted to games in which there is a single round of interaction between the referee and at most $5$ players.
\end{corollary}

The same problem but with no entanglement between the players is contained in $\QCMA$: the players' constant-sized quantum answers can be given as a classical description~\cite{KobMat03JCSS}. It is also known to be $\NP$-hard, even when restricted to classical answers from the players and for constant additive approximations~\cite{Vidick13xor}. However, no \emph{upper bound} is known on the complexity of the problem considered in Corollary~\ref{cor:qma}, which is not even known to be decidable~\cite{ScholzW08,JungeNPPSW11} (and there is no known a priori bound on the amount of entanglement that may be beneficial to the players). 
Corollary~\ref{cor:qma} provides the first indication that entanglement indeed \emph{increases} the verifying power of the referee, at least in the range of inverse-polynomial approximations, showing that unless $\QCMA=\QMA$ the complexity of entangled (quantum) games is strictly \emph{larger} than that of non-entangled (quantum) games.

\paragraph{Consequences for interactive proof systems with entangled provers.}

We can scale up our result to $\QMA_\EXP$, the exponential-witness size version of $\QMA$ (see Section~\ref{sec:prelim} for the definition) to obtain a formal separation between quantum multiprover interactive proof systems with and without entanglement between the provers. Let $\QMIP^*(r,t,c,s)$ be the class of languages that have quantum interactive-proof systems with $r$ provers, $t$ rounds of interaction, completeness $c$ and soundness $s$ (see Section~\ref{sec:prelim} for the complete definition).

\begin{corollary}\label{cor:qmip}
There exists a polynomial $q$ such that 
$$\QMA_{\EXP} \subseteq \QMIP^*(5,1,1-2^{-(q+1)},1-2^{-q})$$
and hence 
$$\QMIP(5,1,1-2^{-(q+1)},1-2^{-q})\subsetneq \QMIP^*(5,1,1-2^{-(q+1)},1-2^{-q})$$
unless $\NEXP=\QMA_{\EXP}$. 
\end{corollary}

The corollary follows from the fact that $\QMIP(5,1,1-2^{-(q+1)},1-2^{-q}) \subseteq \NEXP$~\cite{KobMat03JCSS} and $\NEXP\subseteq \MIP^*(3,1,1,1-1/\poly)$~\cite{IV12} together with the observation $\MIP^*(3,1,1,1-1/\poly) \subseteq \QMIP^*(5,1,1-2^{-(q+1)},1-2^{-q})$. 

We note that even though it is known that $\MIP^* = \QMIP^*$~\cite{ReichardtUV13nature} the above corollary falls short of proving a separation between $\MIP=\NEXP$ and $\MIP^*$. The reason is that the transformation from a $\QMIP^*$ to a $\MIP^*$ protocol in~\cite{ReichardtUV13nature} requires the completeness and soundness parameters of the $\QMIP^*$ protocol to be separated by an inverse polynomial in the input size, whereas our construction only gives an inverse exponential separation. 

\subsection{Proof idea}

Suppose given an instance $H=\sum H_j$ of the local Hamiltonian problem, where each term $H_j$ acts on a subset $S_j=\{i_1,\ldots,i_k\}$ of at most $k$ out of the $n$ qubits. Given an explicit description of $H$, the goal of the verifier is to decide whether there exists a ``proof'' $\ket{\Psi}$ that satisfies most terms $H_j$, i.e. such that the total ``energy'' $\bra{\Psi}H\ket{\Psi}$ is below a certain threshold value. As already mentioned, the main challenge in achieving this is that the verifier will only ever receive, at best, a logarithmic number of qubits of the proof from the provers. Although this easily allows him to estimate the energy $\bra{\Psi}H_j\ket{\Psi}$ of any local term $H_j$, the difficulty is to ensure that the qubits received in response to different queries, associated with different local terms $H_j$, are \emph{globally consistent} --- that they can be ``patched together'' into an actual proof $\ket{\Psi}$ that has low energy with respect to $H$. This difficulty is unique to the case of quantum proofs: if we were working with classical assignments, as explained earlier a simple consistency check would be sufficient to enforce that the provers' answers can be combined into a single assignment satisfying most clauses. But how does one devise a consistency check for quantum proofs, when in general it is not even possible to check whether two quantum states agree locally?\footnote{Pure quantum states $\ket{\Psi}$ and $\ket{\Phi}$ can be compared using the so-called SWAP test. However, for mixed states this test no longer works, and in fact checking consistency of reduced density matrices, even when specified explicitly, is itself a $\QMA$-complete problem~\cite{Liu06consistency}. We refer to~\cite{AharonovAV13qpcp} for more on the difficulties posed by locally checking consistency of quantum states.}

We suggest the following workaround. Our main goal is to ensure that, when a prover is asked for its share of a certain qubit $i_\ell$, or $i_{\ell'}$, of the proof, the actual qubits that it sends back to the verifier in each case do indeed correspond to distinct physical qubits --- that they do not ``overlap'', or even correspond to the same physical qubit, as was the case in our description of a strategy for the frustrated Hamiltonian projecting on overlapping EPR pairs. To enforce this, instead of asking the (honest) provers to directly split the qubits of the original proof between themselves we ask them to share an \emph{encoding} of the proof: each ``logical'' qubit of $\ket{\Psi}$ should be individually encoded into five ``physical'' qubits using a quantum error-correcting code. Each of five provers should then be given one of the five shares associated with each of the original proof's qubits. (Five is the smallest number of qubits for a quantum error-correcting code satisfying the properties we need; although we did not investigate this further a four-qubit error-detecting code may also suffice.)

Given this (presumed) splitting of the proof, we introduce the following protocol, comprised of two tests each applied with probability $1/2$ by the verifier. Observe first that under our encoding it remains easy for the verifier to estimate the energy of any $k$-local term $H_j$: he can ask each of the five provers for its corresponding share of each qubit on which a randomly chosen $H_j$ acts, decode the results, and measure the energy of the resulting qubits with respect to $H_j$. This only requires each prover to send back $k$ qubits to the verifier, and constitutes the first test in our protocol. 

Next consider the following additional test. The verifier chooses a $k$-element subset $S=\{i_1,\ldots,i_k\}$ of $\{1,\ldots,n\}$ uniformly at random. He also selects an index $\ell\in\{1,\ldots,k\}$ at random and asks four out of the five provers (again chosen at random) for their respective share of qubit $i_\ell$ only. To the last prover he asks for its respective shares of all qubits in $S$. (Note that in case $S$ corresponded to the set of $k$ qubits on which a local term $H_j$ acts the last prover cannot distinguish whether it is this test or the first that is being performed, and this will be important for the proof.) The verifier checks that all shares that he received associated with qubit $i_\ell$ lie in the codespace, and rejects the provers if not. 

In this second test the messages sent back by the first four provers only depend on qubit $i_\ell$. The key point is that, informally, given their four respective answer qubits there can exist at most one additional qubit that is entangled with them in a way that completes a valid codeword. Indeed (and again informally), if there were two such qubits it would imply that it is possible for the ``environment'' to entangle itself with a codeword through acting on a single qubit and without being detected by the code --- this possibility is excluded as long as the code is required to correct (or even just detect) all single-qubit errors. Thus this additional test enforces that the qubit sent back by the fifth prover in response to query $i_\ell$ is uniquely specified by the query $i_\ell$; this is acheived by ``locking'' the qubit with the other four provers' answers via the codespace. 

Although the above provides some intuition, proving soundness of the protocol remains technically challenging. We need to show how a complete proof $\ket{\Psi}$ serving as a witness for the energy of the Hamiltonian $H$ can (at least in principle) be reconstructed from prover strategies that are successful in the protocol. Formally each prover's strategy is specified by a pair of unitaries, one for each type of query from the verifier. The difficulty in proving that these unitaries are ``compatible'' and can be composed so as to reconstruct $\ket{\Psi}$ from the provers' entangled registers --- indeed, note a prover may apply an arbitrary transformation to its private space before answering any of the verifier's queries. Our proof specifies an explicit circuit, based on the provers' unitaries, for reconstructing $\ket{\Psi}$ from their initial entangled state. The depth of this circuit is linear in the number of qubits, and it is ultimately this which leads to the polynomial dependence of the soundness parameter on the number of qubits in the proof.


\subsection{Open questions}

Our work gives the first indication that multi-prover interactive proof systems with entangled provers may be strictly more powerful than their purely classical counterparts. Our protocol relies on the use of quantum communication from the provers to the verifier. Although it is known that quantum communication does not increase the power of entangled-prover interactive proof systems, $\QMIP^* = \MIP^*$~\cite{ReichardtUV13nature}, the technique used in~\cite{ReichardtUV13nature} to replace quantum messages by classical ones introduce a polynomial amount of error that, at least if applied na\"ively, would close the completeness/soundness gap of our protocol. We thus leave the possibility of achieving the same results as our ours through a purely classical interaction as an interesting open question. 

The main drawback of our protocol is the scaling of the completeness/soundness gap with the size of the local Hamiltonian instance. The most important question that we leave open for future work is to increase this gap from inverse exponential to inverse polynomial, leading to the inclusion $\QMA_\EXP \subseteq \QMIP^*$. Together with $\QMIP^* = \MIP^*$~\cite{ReichardtUV13nature} such a result would in particular reprove the main result of~\cite{IV12}, and we expect it to pose a significant challenge. Of importance in itself, research on this question could lead to the development of techniques useful to the study of the quantum PCP conjecture~\cite{AharonovAV13qpcp}. To stimulate its exploration we propose that the inclusion  $\QMA_\EXP \subseteq \QMIP^*$ be taken as a second variant of ``quantum PCP conjecture'' --- one we could call the ``interactive-proof QPCP'', in contrast to the ``proof-checking QPCP'' that has so far been the accepted formulation (see e.g. Conjecture~1.4 in~\cite{AharonovAV13qpcp}). No implication is known between the two conjectures; our work provides a first step towards the former, making it potentially more approachable than the latter. 

\paragraph{Acknowledgments.} This work was started while both authors were hosted by the Simons Institute in Berkeley, whose financial support we gratefully acknowledge. The second author is grateful to Dorit Aharonov and Umesh Vazirani for pressing him to expose the question investigated in this paper during an open problems session organized at the institute. Joseph Fitzimons' research is supported by the Singapore National Research Foundation under NRF Award No. NRF-NRFF2013-01. Thomas Vidick's research was supported in part by the Simons Institute and the Ministry of Education, Singapore under the Tier 3 grant MOE2012-T3-1-009.

\section{Preliminaries}\label{sec:prelim}

\paragraph{Notation.} Given a string $x$ we let $|x|$ denote its length. For a set $S$, $|S|$ is its cardinality. For a positive integer $n$ we abbreviate $\{1,\ldots,n\}$ by $[n]$. We use a calligraphic $\mH$ to denote finite-dimensional Hilbert spaces, and roman letters $\rQ,\rR,\ldots$ to denote quantum registers. The Hilbert space associated with register $\rR$ is $\mH_\rR$. We will often, though not always, index kets and bras for quantum states by the names of the registers on which the state lies, e.g. $\ket{\Psi}_{\rQ\rR}$ means that $\ket{\Psi}$ is a bipartite state on $\mH_\rQ\otimes \mH_\rR$.
$\lin{\mH_\rQ,\mH_\rR}$ is the set of all linear maps $\mH_\rQ\to\mH_\rR$. $\pos{\mH}$ is the set of positive operators on $\mH$; $\density{\mH}$ the set of density matrices. Given $\mathcal{F},\mathcal{G}\in\lin{\mH,\mH}$ we let $\mathcal{F}\circ\mathcal{G}$ denote their composition. If there are $s$ such maps $\mathcal{F}_\ell$, we let $\bigcirc_{\ell=1}^{s}\mathcal{F}_\ell := \mathcal{F}_s\circ\cdots\circ \mathcal{F}_1$.

Given two registers $\rQ$ and $\rR$ associated to isomorphic Hilbert spaces $\mH_\rQ$, $\mH_\rR$ respectively we let $\SWAP_{\rQ\rR}$ be the unitary that swaps their contents: for any two orthonormal bases $\ket{u_i}$ for $\mH_\rQ$ and $\ket{v_j}$ for $\mH_\rR$, $\SWAP_{\rQ\rR}= \sum_{i,j} \ket{v_i , u_j}\bra{u_j, v_i}$. 

\paragraph{Complexity classes.}
We give relatively informal definitions of the quantum interactive proof classes considered in this paper. For formal definitions we refer the reader to the book~\cite{KitSheVya02} and the survey~\cite{Watrous2009quantum}.

$\QMA$ is the class of all promise problems $L=(L_{yes},L_{no})$ such that there exists a polynomial $p$ and a quantum polynomial-time verifier $V$ such that:
\begin{itemize}
\item (completeness) For every $x\in L_{yes}$, there exists a state $\ket{\Psi}$ on $p(|x|)$ qubits such that $V(x,\ket{\Psi})$ accepts with probability at least $2/3$,
\item(soundness) For every $x\in L_{no}$ and every $\ket{\Psi}$ on $p(|x|)$ qubits, $V(x,\ket{\Psi})$ accepts with probability at most $1/3$.
\end{itemize}
We further note that using an amplification technique of Marriott and Watrous~\cite{MarriottW05qma} one can show that for any fixed polynomial $q$ the completeness and soundness parameters can be replaced by $1-2^{-q(|x|)}$ and $2^{-q(|x|)}$ respectively without changing the definition of $\QMA$. Furthermore the amplification procedure in~\cite{MarriottW05qma} preserves the witness length, so that the polynomial $p$ does not need to grow if one increases $q$ (only the complexity of the verification procedure increases). We define the exponential-size version of $\QMA$, $\QMA_{\EXP}$, by allowing the witness to be on $2^{p(|x|)}$ qubits and the verifier to run in quantum exponential time.

$\MIP(r,t,c,s)$ is the class of all promise problems $L=(L_{yes},L_{no})$ such that there exists a polynomial $p$ and a classical polynomial-time verifier $V$, interacting with $r$ non-communicating provers through $t$ rounds of interaction in each of which at most $p(|x|)$ bits of communication are exchanged between the verifier and the provers, such that:
\begin{itemize}
\item (completeness) For every $x\in L_{yes}$, there exists a strategy for the provers that is accepted by the verifier with probability at least $c$,
\item(soundness) For every $x\in L_{no}$ any strategy of the provers is accepted by the verifier with probability at most $s$.
\end{itemize}
$\QMIP(r,t,c,s)$ is defined in the same way, except the verifier and communication exchanged are allowed to be quantum. $\MIP^*(r,t,c,s)$ (resp. $\QMIP^*(r,t,c,s)$) is defined as $\MIP(r,t,c,s)$ (resp. $\QMIP(r,t,c,s)$) but the provers are allowed to share an arbitrary entangled state as part of their strategy. (In this paper we only consider protocols for which the number of rounds of interaction is $t=1$.)

It follows from~\cite{BabForLun91CC,KobMat03JCSS} that, for any polynomials $p_1$, $p_2$ and $p_3$,
\begin{align*}
& \MIP(p_1,p_2,2/3,1/3) = \QMIP(p_1,p_2,2/3,1/3) \\
&= \MIP(2,1,1,2^{-p_3}) = \QMIP(2,1,1,2^{-p_3}) = \NEXP.
\end{align*}
In fact,~\cite{KobMat03JCSS} even show that the same equalities hold for $\QMIP^*$ when the provers are limited to a polynomial number of qubits of entanglement. 

\paragraph{The local Hamiltonian problem.}
Let $k$ be a fixed integer and $a,b:\N\to[0,1]$ such that $a(n)<b(n)$ for all integers $n$. The $k$-local Hamiltonian problem (LH) is defined as follows. The input is a classical description of a local Hamiltonian $H = \sum_{i=1}^m H_i \in\lin{\C^{d^n},\C^{d^n}}$ acting on $n$ qudits of dimension $d$ each. Here each $H_i$ is a positive semidefinite matrix of norm at most $1$ acting on at most $k$ out of the $n$ qudits, and can thus be represented by a matrix of dimension $d^k\times d^k$; when we write $H=\sum_i H_i$ we implicitly mean that each $H_i$ should be tensored with the identity acting on the remaining $(n-k)$ qudits. We label  the qudits from $1$ to $n$, and denote by $S_j$ the set of $k$ qudits on which $H_j$ acts. The problem is to determine which of the following two cases holds:
\begin{enumerate}
\item (YES) There exists a $n$-qudit state $\ket{\Gamma}$ such that $\bra{\Gamma}H\ket{\Gamma}\leq am$,
\item (NO) For all states $\ket{\Psi}$, $\bra{\Psi}H\ket{\Psi} \geq bm$. 
\end{enumerate}
Kempe, Kitaev and Regev showed the following:
\begin{theorem}[\cite{KempeKR06lh}]
For any fixed polynomial $q$, there is a polynomial $p$ such that the $k$-local Hamiltonian problem, where the number of qubits $n$ is specified in unary, is $\QMA$-complete for $k=2$, $d=2$, $a=2^{-q(n)}$ and $b=1/p(n)$. 
\end{theorem}
For the case of $\QMA_{\EXP}$ essentially the same construction yields the following (see also~\cite{GottesmanI09tiling}): 
\begin{theorem}[\cite{KempeKR06lh}]
For any fixed polynomial $q$, there is a polynomial $p$ such that the $k$-local Hamiltonian problem, where the number of qubits $N$ is specified in binary (hence can be exponential in the input size), is $\QMA_{\EXP}$-complete for $k=2$, $d=2$, $a=2^{-q(N)}$ and $b=1/p(N)$. 
\end{theorem}

\paragraph{Error-correcting codes.}
Our protocol relies on the use of a quantum error-correcting code $C$ that has the following properties:
\begin{itemize}
\item $C$ encodes $1$ logical qubit into $r$ physical qubits.
\item $C$ detects and corrects all single-qubit Pauli errors on a single qubit.
\item The reduced density matrix of any codewords in $C$ on a single qubit is the totally mixed state $\Id/2$.
\end{itemize}
An example of a code satisfying all three conditions for $r=5$ and $e=1$ is the $5$-qubit stabilizer code~\cite{bennett1996mixed,laflamme1996perfect}. 
Given $r$ single-qubit registers $\rR_1,\ldots,\rR_r$ we let $\DEC_{\rR_1\cdots\rR_r}:\density{(\C^2)^{\otimes r}} \to\density{ \C^2}$ be the completely positive trace-preserving (CPTP) map corresponding to the decoding operation. We also let $\CHECK_{\rR_1\cdots\rR_r} \in \pos{(\C^2)^{\otimes r}}$ be the projection onto the code space.

\section{Proof of Theorem~\ref{thm:main}}

\begin{figure}
\begin{protocol*}{Protocol~$P$}
\begin{step}
\item[]  Let $H=\sum_{i=1}^m H_i$ be an instance of the $k$-local Hamiltonian problem given as input, and $n$ the number of qubits on which $H$ acts. Let $C$ be an error-correcting code which encodes $1$ logical qubit into $r$ physical qubits and satisfies the three conditions described at the end of Section~\ref{sec:prelim}.  

The verifier performs each of the following tests with probability $1/2$ each:
\begin{step}
\item\label{step:test1} Select a $j\in[m]$ uniformly at random, and let $S_j\subseteq[n]$ be the set of $k$ qubits on which the local term $H_j$ acts. Ask the provers for their respective share of all qubits in $S_j$. Upon receiving the shares, apply the decoding map independently to each of the $k$ groups of $r$ shares and measure the resulting state using $\{H_j,\Id-H_j\}$. Reject if the outcome is `$H_j$'. 
\item\label{step:test2} Select a qubit $i\in [n]$ uniformly at random, and a set $S\subseteq [n]$ uniformly at random among all sets of size $k$ that contain $i$. With probability 1/2, ask one of the provers at random for his share of all qubits in $S$, and the remaining $r-1$ provers for their respective share of the $i$-th qubit only. With probability 1/2, ask all provers for their respective share of the $i$-th qubit. In both cases, verify that all provers' shares of the $i$-th qubit together lie in the codespace. Reject if not. 
\end{step}
\end{step}
\end{protocol*}
\caption{Protocol for the verification of an instance of the local Hamiltonian problem.}
\label{fig:protocol}
\end{figure}

In this section we prove Theorem~\ref{thm:main}. The protocol is described in Figure~\ref{fig:protocol}. The first two properties claimed in the theorem, on the structure of the protocol, are clear: there is a single round of interaction, and using the $5$-qubit stabilizer code for $C$ the protocol can be executed with $r=5$ provers. Messages from the verifier to the provers are either the label of a qubit or the description of a set of size $k$, which require $O(\log n)$ bits to specify. Messages from any prover to the verifier are either $1$ or $k$ qubits. In Section~\ref{sec:completeness} we establish the completeness property of the protocol; soundness is proved in Section~\ref{sec:soundness}. 

\subsection{Completeness analysis}\label{sec:completeness}

\begin{lemma}
Suppose that there exists a state $\ket{\Gamma}$ such that $\bra{\Gamma}H\ket{\Gamma}\leq am$. Then there exists a strategy for the provers in Protocol~$P$ that is accepted with probability at least $1-a/2$.  
\end{lemma}

\begin{proof}
We describe a strategy for the provers. Let $\ket{\Gamma}$ be such that $\bra{\Gamma}H\ket{\Gamma}\leq am$. Before the protocols start, the provers generate a shared entangled state $\ket{\Psi}$ over $rn$ qubits by independently encoding each qubit of $\ket{\Gamma}$ into $r$ qubits using the code $C$ prescribed by the protocol. Each of the $r$ provers keeps $n$ qubits of $\ket{\Psi}$, corresponding to a share of each of the encoded qubits of $\ket{\Gamma}$. When asked for its share of any set of qubits, the prover complies and sends it to the verifier. It is clear that this strategy is accepted with probability $1$ in item~\ref{step:test2}, and with probability 
$$\frac{1}{m}\sum_{i=1}^m \,\bra{\Gamma}(\Id-H_i)\ket{\Gamma}\,\geq\,1-a$$
 in item~\ref{step:test1}. Using that each test is performed with probability $1/2$, the overall success probability for the strategy is at least $1-a/2$. 
\end{proof}

\subsection{Soundness analysis}\label{sec:soundness}

In this section we analyze the soundness of protocol~$P$. In section~\ref{sec:def-strategy} we introduce the notation used to describe the most general strategy that the provers may employ in the protocol. In section~\ref{sec:strat-analysis} we show that, provided that all eigenvalues of $H$ are larger than some inverse polynomial, any strategy for the provers is rejected by the verifier with inverse polynomial probability. 

\subsubsection{The provers' strategies}\label{sec:def-strategy}

We denote an arbitrary strategy for the $r$ provers in protocol~$P$ via  a triplet $(U_i^j,V_S^j,\ket{\Psi})$ (or sometimes $(U_i^j,V_S^j,\rho)$). Here $\ket{\Psi}$ (or $\rho$) denotes the initial $r$-partite entangled state shared by the provers, and $U_i,V_S$ the unitaries that they apply upon receiving questions $i,S$ respectively. More precisely, in the protocol a prover is asked two types of questions. Either it is asked for a single qubit $i$, in which case we call the unitary $U_i^t$ (where $t$ indexes the prover), or it is asked  for a set of $k$ qubits $S$, in which case we call the unitary $V_S^t$. We sometimes omit the superscript $t$, as the labeling of the provers will often be clear from context.  We denote the associated completely positive trace-preserving (CPTP) maps by $\mathcal{U}_i^t:\sigma\mapsto U_i^t \sigma (U_i^t)^\dagger$ and $\mathcal{V}_{S}^t:\sigma\mapsto V_{S}^t \sigma (V_{S}^t)^\dagger$.

For $t\in [r]$ we write $\rP^t$ for the register containing the $t$-th prover's share of $\ket{\Psi}$. After application of the unitary $U_i^t$ or $V_S^t$, we relabel registers associated to the prover as $\rS^t,\rQ_1^t,\ldots,\rQ_n^t$. Here the $n$ registers $\rQ_1^t,\ldots,\rQ_n^t$ are each single-qubit registers such that register $\rQ_i^t$ (resp. registers $\rQ_{i_1}^t\cdots \rQ_{i_k}^t$) is sent back to the verifier when the prover is asked for qubit $i$ (resp. set of qubits $S=\{i_1,\ldots,i_k\}$). Note that all registers $\rQ_i^t$ may not exist simultaneously; which ones do depends on the unitary $U_i^t$ or $V_S^t$ that was applied. The remaining register $\rS^t$ is an auxiliary register of arbitrary dimension. In addition, for each prover $t\in\{1,\ldots,r\}$ we introduce $2n$ auxiliary registers $\rR_1^t,\ldots,\rR_n^t$ and $\overline{\rR}_1^t,\ldots,\overline{\rR}_n^t$, and define
\begin{equation}\label{eq:defpsi}
\ket{\tilde{\Psi}} := \ket{\Psi} \bigotimes_{t=1}^r \bigotimes_{i=1}^n \frac{1}{\sqrt{2}}\big(\ket{00}_{\rR_i^t\overline{\rR}_i^t}+\ket{11}_{\rR_i^t\overline{\rR}_i^t}\big),
\end{equation}
i.e. $\ket{\tilde{\Psi}}$ is $\ket{\Psi}$ adjoined with $n$ EPR pairs for each prover, created in the auxiliary registers. We write $\rho = \ket{\Psi}\bra{\Psi}$ and $\tilde{\rho}=\ket{\tilde{\Psi}}\bra{\tilde{\Psi}}$. See Figure~\ref{fig:registers} for a summary of our nomenclature for registers. We will often abbreviate $\rQ_i$ for the union of the $\rQ_i^j$, $j\in [r]$, and write $\rQ_i^{\neq t}$ for the union of all $\rQ_i^j$ for $j\in[r]\backslash\{t\}$. 

\begin{figure}
\begin{center}
\begin{tabular}{r|c|l}
&Register &  Use  \\ \hline
Before application &  $\rP^t$ & Prover $t$'s register in state $\ket{\Psi}$ \\
 of $U_i$, $V_{S}$. & & \\ \hline
After application &   $\rQ_i^t$ & Register sent back by prover $t$ if asked for the $i$-th qubit.\\ 
 of $U_i$, $V_{S}$ &$\rS^t$ & Prover $t$'s remaining private registers.\\\hline
Auxiliary & $\rR_i^t$, $\overline{\rR_i^t}$ & Initialized as an EPR pair. \\
registers & &
\end{tabular}
\end{center}
\caption{Notation for the provers' registers.}
\label{fig:registers}
\end{figure}

We introduce a new set of unitaries which act on a prover's share of $\ket{\tilde{\Psi}}$ as
\begin{equation}\label{eq:def-cd}
C^t_i\,:=\, (U^t_i)^\dagger (\SWAP_{\rQ_i^t\rR_i^t}\otimes \Id) U_i^t\qquad\text{and}\qquad D_{i,S}^t\,:=\, (V_S^t)^\dagger (\SWAP_{\rQ_i^t\rR_i^t} \otimes \Id) V_S^t,
\end{equation}
where $U_i^t$ and $V_S^t$ are implicitly tensored with the identity on the auxiliary registers. We denote the associated CPTP maps by $\cC_i^t:\sigma\mapsto C_i^t \sigma (C_i^t)^\dagger$ and $\cD_{i,S}^t:\sigma\mapsto D_{i,S}^t \sigma (D_{i,S}^t)^\dagger$. In order not to overload the notation we often do not specify precisely on which registers the identity acts (sometimes we even omit the symbol $\Id$ altogether), as it should always be clear from context.
In words, $C_i^t$ corresponds to applying $U_i^t$, swapping the register $\rQ_i^t$ containing the output qubit with the $i$-th ancilla register $\rR_i^t$, and applying $(U_i^t)^\dagger$. For $i\in S$, $D_{i,S}^t$ is defined as $C_i^t$ but from the unitary $V_S^t$ instead of $U_i^t$, while still swapping the output qubit in register $\rQ_i^t$ only (and not the others). For any subset $T\subseteq S$ we define $D_{T,S}^t$ in the same ways as $D_{i,S}^t$ except all qubits in the subset $T$ are swapped out; in particular $D_{\{i\},S}^t = D_{i,S}^t$ and $D_{\emptyset,S}=\Id$. The following observation, which follows from $V_S^t(V_S^t)^\dagger = \Id$, will be useful:
\begin{equation}\label{eq:grow-T}
\forall T\subset S,\,\forall i\in S\backslash T,\qquad  D^t_{i,S} D^t_{T,S} \,=\, D^t_{T,S} D^t_{i,S}\,=\,D^t_{T\cup\{i\},S}.
\end{equation}
Since $\SWAP = \SWAP^\dagger$ it also holds that $(C_i^t)^\dagger = C_i^t$ and $(D_{T,S}^t)^\dagger = D_{T,S}^t$. 

Finally, we define an $n$-qubit mixed state 
\begin{equation}\label{eq:def-witness}
\sigma := \Big(\bigotimes_{i=1}^n \DEC_{\rR_i^1\cdots\rR_i^r} \Big)\Big( \Tr_{\cup_t((\cup_i\overline{\rR_i^t}\rQ_i^t) \rS^t) }\Big(\Big(\bigotimes_{t=1}^{r} C_n^t\cdots  C_2^t  C_1^t \Big)\ket{\tilde{\Psi}}\bra{\tilde{\Psi}} \Big(\bigotimes_{t=1}^{r} (C_1^t)^\dagger\cdots  (C_n^t)^\dagger\Big)\Big)\Big),
\end{equation}
i.e. $\sigma$ is the state obtained by, first applying unitaries $C_1^t,\ldots,C_n^t$, for $t=1,\ldots,r$, to the original state $\ket{\Psi}$ and the auxiliary registers (initialized as EPR pairs), then tracing out all but the $nr$ auxiliary registers $\rR_1^t,\ldots,\rR_n^t$ for $t=1,\ldots,r$, and finally applying the decoding map for code $C$ independently to each group of $r$ auxiliary registers $\rR_i^1\cdots \rR_i^r$.

\subsubsection{Analysis of the strategy}\label{sec:strat-analysis}

In this section we prove the following lemma, which establishes soundness of protocol $P$.  

\begin{lemma}\label{lem:soundness}
There exists a universal constant $c_3>0$ (depending on $k$ only) such that the following holds. Suppose a strategy for the provers is accepted with probability at least $1-\eps$ in each of the tests of protocol~$P$, for some $\eps>0$. Then the state $\sigma$ defined in~\eqref{eq:def-witness} satisfies $\frac{1}{m}\Tr(H\sigma) =O(n^{c_3}\eps)$. 
\end{lemma}

The proof of the lemma follows from a sequence of claims. The first draws a useful consequence of the condition that the provers succeed in test~\ref{step:test2} with high probability. 

\begin{claim}\label{claim:cidj}
Suppose the strategy $(U_i^j,V_{S}^j,\ket{\Psi})$ succeeds in test~\ref{step:test2} with probability at least $1-\eps$. For any $t\in [r]$, $i\in [n]$ and $S\subseteq [n]$ of cardinality $k$ such that $i\in S$,
\begin{equation}\label{eq:cidj1}
\big\| (C_i^t - D_{i,S}^t)\otimes \Id \ket{\tilde{\Psi}}\big\|^2 \,=\, O\big(n^k\eps\big),
\end{equation}
where $\ket{\tilde{\Psi}}$ is defined from $\ket{\Psi}$ in~\eqref{eq:defpsi}. Furthermore, for any set $S'\subseteq [n]$ of cardinality $k$ and $T\subseteq S\cap S'$, 
\begin{equation}\label{eq:cidj2}
 \big\| (D_{T,S}^t - D_{T,S'}^t)\otimes \Id \ket{\tilde{\Psi}}\big\|^2 \,=\, O\big(n^k\eps\big).
\end{equation}
\end{claim}

\begin{proof}
For any $t\in[r]$, $i\in [n]$ and set $S\subseteq[n]$ such that $i\in S$ let 
\begin{equation}\label{eq:def-varphi}
\ket{\varphi_{i}}\,:=\,\bigotimes_{p=1}^r C_{i}^{p} \ket{\tilde{\Psi}}\qquad\text{and}\qquad \ket{\varphi_{i,S}}\,:=\,D_{i,S}^t\Big(\bigotimes_{p\neq t}  C_{i}^{p}\Big) \ket{\tilde{\Psi}},
\end{equation}
where for ease of notation the dependence on $t$ of $\ket{\varphi_{i}}$ and $\ket{\varphi_{i,S}}$ is left implicit. By definition, this strategy's success probability in test~\ref{step:test2} of the protocol is exactly 
\begin{align*}
\frac{1}{r}\sum_{t=1}^r \frac{1}{n}\sum_{i=1}^n \, \frac{1}{\binom{n-1}{k-1}} \sum_{S:\, i\in S} \frac{1}{2}\Big(&\bra{\Psi}\Big(\bigotimes_{p=1}^r U_i^p\Big)^\dagger \CHECK_{\rQ_i^1\cdots\rQ_i^r}\Big(\bigotimes_{p=1}^r U_i^p\Big) \ket{\Psi}\\
&+\bra{\Psi}\Big(V_{S}^t \bigotimes_{p\neq t} U_i^p \Big)^\dagger\CHECK_{\rQ_i^1\cdots\rQ_i^r}\Big(V_{S}^t\bigotimes_{p\neq t} U_i^p\Big)\ket{\Psi}\Big).
\end{align*}
Let $\CK_i:=\CHECK_{\rR_i^1\cdots\rR_i^r}$. Given the definition of $C_i^t$ and $D_{i,S}^t$ in~\eqref{eq:def-cd}, success $1-\eps$ in test~\ref{step:test2} of the protocol can be rewritten as 
\begin{equation}\label{eq:cidj-2}
\frac{1}{n}\sum_{i=1}^n \,\frac{1}{\binom{n-1}{k-1}} \sum_{S:\,i\in S}  \frac{1}{2}\Big(\bra{\varphi_{i}}\CK_i\ket{\varphi_{i}}+\bra{\varphi_{i,S}}\CK_i\ket{\varphi_{i,S}}\Big)\,\geq \,1-\eps_t,
\end{equation}
where the $\eps_t$ satisfy $(1/r)(\eps_1+\cdots+\eps_r) = \eps$. 
Decompose the action of the unitary $D_{i,S}^t(C_i^t)^\dagger$ as
\begin{equation}\label{eq:cidj-1}
D_{i,S}^t(C_i^t)^\dagger \,=\, \Id_{\rR_i^t}\otimes W_{i,S}^1 + X_{\rR_i^t}\otimes W_{i,S}^2 + Y_{\rR_i^t}\otimes W_{i,S}^3 + Z_{\rR_i^t}\otimes W_{i,S}^4,
\end{equation}
where the Pauli operators $\{\Id,X,Y,Z\}$ act on the $i$-th auxiliary register $\rR_i^t$ associated with the $t$-th prover, and the $W_{i,S}^\ell$ are arbitrary operators (not necessarily unitary) of norm at most $1$ acting on the remaining registers $\rQ_1^t\cdots \rQ_n^t$ and $\rS^t$. Note that both $C_i^t$ and $D_{i,S}^t$ are such that $\Tr_{\rR_i^t}(C_i^t)= \Tr_{\rR_i^t}(D_{i,S}^t)=\Id_{\rQ_1^t\cdots \rQ_n^t\rS^t}$, hence $W_{i,S}^1 = \Id$. Let 
$$\ket{\varphi_i^s} \,:=\, \big(\CK_i\otimes \Id\big)\ket{\varphi_i}\qquad\text{and}\qquad\ket{\varphi_i^f} \,:=\, \big(\big(\Id-\CK_i\big)\otimes \Id\big)\ket{\varphi_i},$$
so that $\ket{\varphi_i} = \ket{\varphi_i^s}+\ket{\varphi_i^f}$. By assumption the code $C$ corrects all single-qubit Pauli errors, and since by definition the reduced density of $\ket{\varphi_i^s}$ on registers $\rR_i^1\cdots\rR_i^r$ is in the codespace, for  any single-qubit Pauli error $E_{R_i^t}\in\{X,Y,Z\}$ acting on register $\rR_i^t$,
\begin{equation}\label{eq:cidj-3}
\CHECK_{\rR_i^1\cdots\rR_i^r}\big(E_{\rR_i^t}\otimes \Id \big)\ket{\varphi_i^s}\,=\,0.
\end{equation}
As a consequence, starting from the definition of $\ket{\varphi_{i,S}}$ and using the decomposition~\eqref{eq:cidj-1} we get
\begin{align}
\CK_i\ket{\varphi_{i,S}} &= \CK_i\cdot\Big(D_{i,S}^t \big(C_i^t\big)^\dagger \otimes \Id \Big)\ket{\varphi_i}\notag\\
&=\CK_i\cdot\Big( \Id_{\rR_i^t}\otimes \Id + X_{\rR_i^t}\otimes W_{i,S}^2 + Y_{\rR_i^t}\otimes W_{i,S}^3  + Z_{\rR_i^t}\otimes W_{i,S}^4 \Big)\ket{\varphi_i}\notag\\
&=  \CK_i\otimes \Id  \ket{\varphi_i}+ \big(\CK_i\cdot X_{\rR_i^t}\otimes W_{i,S}^2  + \CK_i \cdot Y_{\rR_i^t}\otimes W_{i,S}^3 + \CK_i\cdot Z_{\rR_i^t}\otimes W_{i,S}^4\big)\ket{\varphi_i^s}\notag\\
&\qquad+ \big(\CK_i \cdot X_{\rR_i^t}\otimes W_{i,S}^2 + \CK_i\cdot Y_{\rR_i^t}\otimes W_{i,S}^3  + \CK_i \cdot Z_{\rR_i^t}\otimes W_{i,S}^4 \big)\ket{\varphi_i^f}\notag\\
&= \CK_i\otimes \Id  \ket{\varphi_i}+ \big(\CK_i\cdot X_{\rR_i^t}\otimes W_{i,S}^2  +\CK_i\cdot Y_{\rR_i^t}\otimes W_{i,S}^3 + \CK_i\cdot Z_{\rR_i^t}\otimes W_{i,S}^4\big) \ket{\varphi_i^f},\label{eq:cidj-4}
\end{align}
where the last equality follows from~\eqref{eq:cidj-3} and the fact that the $W_{i,S}^j$ do not act on $\rR_i^t$. Eq.~\eqref{eq:cidj-2} implies that both
\begin{equation}\label{eq:cidj-4a}
\frac{1}{n}\sum_{i=1}^n \,\frac{1}{\binom{n-1}{k-1}} \sum_{S:\,i\in S} \big\|\ket{\varphi_i^f}\big\|^2\,\leq\,2\eps_t
\end{equation}
and
\begin{equation}\label{eq:cidj-4b}
\frac{1}{n}\sum_{i=1}^n \,\frac{1}{\binom{n-1}{k-1}} \sum_{S:\,i\in S} \big\|(\Id-\CHECK_{\rR_i^1\cdots\rR_i^r})\ket{\varphi_{i,S}}\big\|^2\,\leq\, 2\eps_t,
\end{equation}
where we used that $\CHECK_{\rR_i^1\cdots\rR_i^r}$ is a projection. Using the triangle inequality as
\begin{align*}
 \big\|\ket{\varphi_{i,S}} -   \ket{\varphi_i}\big\|&\leq  \big\|\ket{\varphi_{i,S}} -\CK_i \ket{\varphi_{i,S}}\big\| + 
\big\|\CK_i\ket{\varphi_{i,S}} -\CK_i  \ket{\varphi_i}\big\| + \big\|\CK_i\ket{\varphi_{i}} -  \ket{\varphi_i}\big\| 
\end{align*}
we get
\begin{equation}\label{eq:cidj-5}
\frac{1}{n}\sum_{i=1}^n \,\frac{1}{\binom{n-1}{k-1}} \sum_{S:\,i\in S}\, \big\|\ket{\varphi_{i,S}} -  \ket{\varphi_i}\big\|^2\,\leq\, 3\big(2\eps_t+9\cdot2\eps_t + 2\eps_t) \,=\,O(\eps_t),
\end{equation}
where the first bound is obtained from~\eqref{eq:cidj-4b}, the second from~\eqref{eq:cidj-4},~\eqref{eq:cidj-4a} and $\|\CK_i\|,\|W_{i,j}^\ell\|\leq 1$, and the third from the definition of $\ket{\varphi_i^f}$ and~\eqref{eq:cidj-4a}. Recalling the definition of $\ket{\varphi_{i}}$ and $\ket{\varphi_{i,S}}$ in~\eqref{eq:def-varphi},~\eqref{eq:cidj1} is proved by noting that the operator $(\Id \otimes_{p\neq t}C_i^p)$ is unitary and hence its application does not modify the Euclidean norm. 

The proof of~\eqref{eq:cidj2} follows the same steps. Defining vectors $\ket{\varphi_{T,S}}$ and $\ket{\varphi_{T,S'}}$ and using that~\eqref{eq:cidj-2} is satisfied for every $i\in T$ we can decompose $D_{T,S}^t (D_{T,S'}^t)^\dagger$ as in~\eqref{eq:cidj-1}, except now the decomposition involves all $|T|$-qubit Pauli operators on registers $\rR_i^t$ for $i\in T$. The different qubits are checked independently, and we can define $\ket{\varphi_{T,S'}^s} := (\otimes_{i\in T} \CK_i)\ket{\varphi_{T,S'}}$. The remainder of the derivation follows the same steps, leading to~\eqref{eq:cidj2} (where factors polynomial in $k$ are hidden in the $O(\cdot)$ notation, using that $k$ is a constant independent of $n$).
\end{proof}

For any $i\in [n]$ let $ \mathcal{F}_{i}$ be the completely positive trace non-increasing map, acting on all provers' registers, defined by
\begin{equation}\label{eq:def-fi}
 \mathcal{F}_{i}: \,\sigma\,\mapsto\, \Big(\big(\bigotimes_{j=1}^r X_{i}^j\big)^\dagger \CK_{\rQ_i^1\cdots\rQ_i^r}\big(\bigotimes_{j=1}^r X_{i}^j\big)\Big)\,\sigma\,\Big( \big(\bigotimes_{j=1}^r X_{i}^j\big)^\dagger \CK_{\rQ_i^1\cdots\rQ_i^r}\big(\bigotimes_{j=1}^r X_{i}^j\big)\Big)^\dagger.
\end{equation}
Here we use the symbol $X_i^j$ to represent any of $U_i^j$ or $V_S^j$ for any $S$ containing $i$; we leave the dependence of $\mathcal{F}_i$ on the choice of $X_i^j$ implicit as all bounds proved will hold irrespective of that choice. We also write $\mathcal{X}_i^j : \sigma\to X_i^j\sigma (X_i^j)^\dagger$ for the CPTP map associated with $X_i^j$. 
Note that, in addition to the presence of the $\CK$ operator, the difference between the maps $\mathcal{F}_i$ and e.g. $\otimes_j \mathcal{C}_i^j$ is that in the former the $t$ registers $\rQ_i$ and $\rR_i$ are not swapped; in particular $\mathcal{F}_i$ acts as identity on $\rR_i$. 

Our second claim shows that the property that the qubits extracted from the provers' strategies through the maps $X_i^j$ are in the codespace remains preserved even after many layers of application of the $\mathcal{F}_i$.

\begin{claim}\label{claim:check}
Suppose the strategy $(U_i^j,V_{S}^j,\ket{\Psi})$ succeeds in test~\ref{step:test2} with probability at least $1-\eps$. Let $s$ be an integer and $i_1,\ldots,i_s\in [n]$. Then
\begin{equation}\label{eq:check} \Tr\Big( \Big(\bigcirc_{\ell=1}^{s} \mathcal{F}_{i_\ell}\Big)(\tilde{\rho}\big)\Big)\,=\, 1-O\big(sn^k\eps).
\end{equation}
\end{claim}

\begin{proof}
We prove~\eqref{eq:check} by induction on $s$. For $s=1$ it follows immediately by first applying  Claim~\ref{claim:cidj} (at most) $r$ times to replace each $X_{i_1}^j$ in the definition of $\mathcal{F}_{i_1}$ by $U_{i_1}^j$, and then using the assumption of success in the test, which ensures that the extracted qubits are close to the code space. Suppose~\eqref{eq:check} verified for some $s$, and let $K$ be the constant implicit in the $O(\cdot)$ notation; we show it for $s+1$. Writing $\CK_{i_1} = \Id - (\Id-\CK_{i_1})$ and using that $\mathcal{F}_{i_1}$ reduces to the identity once the operator $\CK_{\rQ_{i_1}^1\cdots \rQ_{i_1}^r}$ is removed,
\begin{align*}
 \Tr\Big( \Big(\bigcirc_{\ell=1}^{s+1} \mathcal{F}_{i_\ell}\Big)(\tilde{\rho}\big)\Big) &=  \Tr\Big( \Big(\bigcirc_{\ell=2}^{s+1} \mathcal{F}_{i_\ell}\Big)(\tilde{\rho}\big)\Big) \\
&\qquad - \Tr\Big( \Big(\bigcirc_{\ell=2}^{s+1} \mathcal{F}_{i_\ell}\Big)\circ \Big( \bigotimes_{j=1}^r \mathcal{X}_{i_1}^j \Big)^\dagger\Big((\Id-\CK_{\rQ_{i_1}^1\cdots \rQ_{i_1}^r})\Big( \bigotimes_{j=1}^r \mathcal{X}_{i_1}^j \Big) (\tilde{\rho}\big)\CK_{\rQ_{i_1}^1\cdots \rQ_{i_1}^r}\Big)\Big) \\
&\geq 1-Ksn^k\eps -  \Tr\Big( (\Id-\CK_{\rQ_{i_1}^1\cdots \rQ_{i_1}^r})\Big( \bigotimes_{j=1}^r \mathcal{X}_{i_1}^j \Big) (\tilde{\rho}\big)\Big) \\
&\geq 1-Ksn^k\eps - O\big(n^k\eps\big),
\end{align*}
where the first inequality uses the induction hypothesis for the first term, and that the $\mathcal{F}_{i_1}$ are trace non-increasing for the second, and the last follows from the case $s=1$ of~\eqref{eq:check}. Provided $K$ is chosen large enough this establishes the induction step and proves the claim. 
\end{proof}

The next claim has a similar flavor as the previous one, that the qubits extracted from the provers' strategies lie in the codespace is preserved even after application of a sequence of maps $\mathcal{C}_i^t$ or $\mathcal{D}_{i,S}^t$ on one of the provers' registers. 

\begin{claim}\label{claim:collapse}
There exists a constant $c_1>0$ depending on $k$ only such that the following holds. Suppose the strategy $(U_i^j,V_{S}^j,\ket{\Psi})$ succeeds in test~\ref{step:test2} with probability at least $1-\eps$. Let $s$ be an integer and $i_1,\ldots,i_s\in [n]$. Then for any $t\in [r]$ and choice of $\mathcal{Y}_{i_\ell}^j \in \{\mathcal{C}_{i_\ell}^j,\,\mathcal{D}_{i_\ell,S_{\ell}}^j|\,i_\ell\in S_\ell\}$ for $j\in[r]$ and $\ell\in[s]$,
\begin{equation}\label{eq:collapse}
\Tr\Big( \CK_{i_s} \Big(\Big( \bigcirc_{\ell=1}^{s} \mathcal{Y}_{i_\ell}^t \Big)\bigotimes_{j\neq t} \mathcal{Y}_{i_s}^j \Big) \big(\tilde{\rho}\big)\Big) = 1-O\big(s^2 n^{c_1}\eps\big).
\end{equation}
\end{claim}

\begin{proof}
For the proof we show that the following holds by downwards induction on $s'$ from $s$ to $1$:
\begin{equation}\label{eq:collapse-1}
\Tr\Big( \CK_{i_s}\Big(\Big( \bigcirc_{\ell=s'}^{s} \mathcal{Y}_{i_\ell}^t \Big)\bigotimes_{j\neq t} \mathcal{Y}_{i_s}^j \Big) \circ\Big( \bigcirc_{\ell=1}^{s'-1} \mathcal{F}_{i_\ell} \Big)\big(\tilde{\rho}\big)\Big) = 1-O\big(ss' n^{c_1}\eps\big).
\end{equation}
Eq.~\eqref{eq:collapse-1} for $s'=s$ is equivalent to~\eqref{eq:check}, so Claim~\ref{claim:check} proves the base case for the induction. If $s'=1$ it reduces to~\eqref{eq:collapse}, which is what we need to prove. Assume thus~\eqref{eq:collapse-1} verified for $s'+1$, and prove it for $s'$. 
 By~\eqref{eq:check} applied with $s=s'$ we see that the strategy $(U_i^j,V_{S}^j,\tilde{\rho}_{s'})$, where $\tilde{\rho}_{s'} =  \bigcirc_{\ell=1}^{s'} \mathcal{F}_{i_\ell} (\rho)$ (it will be more convenient to leave $\tilde{\rho}_{s'}$ unnormalized) succeeds in test~\ref{step:test2} with probability at least $1-O(sn^{k}\eps)$. Here in the definition of $\mathcal{F}_{i_{s'}}$ we define $\mathcal{X}_{i_{s'}}^t$ as $ \mathcal{U}_{i_{s'}}^t$ if  $\mathcal{Y}_{i_{s'}}^t =  \mathcal{C}_{i_{s'}}^t$ and $ \mathcal{V}_{S'}^t$ if  $\mathcal{Y}_{i_{s'}}^t =  \mathcal{D}_{i_{s'},S'}^t$; the remaining $\mathcal{X}_{i_{\ell}}^j$ can be chosen arbitrarily.
Applying Claim~\ref{claim:cidj} we get 
\begin{equation}\label{eq:collapse-2}
\Big\|\Big( \bigcirc_{\ell=s'+1}^{s} \mathcal{Y}_{i_\ell}^t \Big)\bigotimes_{j\neq t} \mathcal{Y}_{i_s}^j \Big)\big(\tilde{\rho}_{s'}\big) - \Big( \bigcirc_{\ell=s'+1}^{s} \mathcal{Y}_{i_\ell}^t \Big)\bigotimes_{j\neq t} \mathcal{D}_{i_s,S}^j \Big)\big(\tilde{\rho}_{s'}\big)\Big\|_1 = O(sn^{2k}\eps),
\end{equation}
where for $S$ we choose any set containing both $i_{s'}$ and $i_{s}$. Applying the claim once more, this time starting from the strategy $(U_i^j,V_{S}^j,\tilde{\rho}_{s'-1})$, we have
\begin{equation}\label{eq:collapse-3}
\Big\| \mathcal{F}_{i_{s'}}\big(\tilde{\rho}_{s'-1}\big) -  \Big(\mathcal{X}_{i_{s'}}^t\bigotimes_{j\neq t}^r \mathcal{V}_{S}^j\Big)^\dagger \Big(\CK_{\rQ_{i_{s'}}^1\cdots \rQ_{i_{s'}}^r}\Big( \mathcal{X}_{i_{s'}}^t\bigotimes_{j\neq t}^r \mathcal{V}_{S}^j \Big)\big(\tilde{\rho}_{s'-1}\big)\CK_{\rQ_{i_{s'}}^1\cdots \rQ_{i_{s'}}^r}\Big)\Big\|_1 = O(sn^{2k}\eps).
\end{equation}
Combining~\eqref{eq:collapse-2} and~\eqref{eq:collapse-3} by the triangle inequality and evaluating the overlap with $\CK_{i_s}$ we get
\begin{align}
\Big|\Tr &\Big( \CK_{i_{s}}\Big( \Big( \bigcirc_{\ell=s'+1}^{s} \mathcal{Y}_{i_\ell}^t \Big)\bigotimes_{j\neq t} \mathcal{Y}_{i_s}^j \Big) \big(\tilde{\rho}_{s'}\big)\Big)\notag\\
& - \Tr\Big(\CK_{i_s} \Big( \bigcirc_{\ell=s'+1}^{s} \mathcal{Y}_{i_\ell}^t \circ\big(\mathcal{X}_{i_{s'}}^t\big)^\dagger \Big) \Big(\CK_{\rQ_{i_{s'}}^1\cdots \rQ_{i_{s'}}^r}\Big( \mathcal{X}_{i_{s'}}^t\bigotimes_{j\neq t}^r \mathcal{V}_{S}^j \Big)\big(\tilde{\rho}_{s'-1}\big)\CK_{\rQ_{i_{s'}}^1\cdots \rQ_{i_{s'}}^r}\Big)\Big)\Big| = O(sn^{2k}\eps),\label{eq:collapse-4}
\end{align}
where we also used the definition of $\mathcal{D}_{i_{s'},S}^j$ to simplify the successive application of unitaries $V_S^j$ and $(V_S^j)^\dagger$ on provers $j\neq t$ in the second term above. The fact that every codeword of the code $C$ has a reduced density on any single qubit that is totally mixed implies that  if we trace out registers $\rQ_{i_{s'}}^{\neq t}$ and $\overline{\rR}_{i_{s'}}^t$ in the (unnormalized) density $\tilde{\sigma}_{s'} = \CK_{i_{s'}}(\mathcal{X}_{i_{s'}}^t \bigotimes_{j\neq t} \mathcal{V}_{S}^j )(\tilde{\rho}_{s'-1})\CK_{i_{s'}}$, the registers ${\rR}_{i_{s'}}^t {\rQ}_{i_{s'}}^t$ are jointly in the totally mixed state. Swapping the two registers thus leaves the state invariant, and from~\eqref{eq:collapse-4} we get
\begin{align}\label{eq:collapse-5}
\Big|\Tr &\Big( \CK_{i_{s}}\Big( \Big( \bigcirc_{\ell=s'+1}^{s} \mathcal{Y}_{i_\ell}^t \Big)\bigotimes_{j\neq t} \mathcal{Y}_{i_s}^j \Big) \big(\tilde{\rho}_{s'}\big)\Big)\notag\\
& - \Tr\Big(\CK_{i_s} \Big( \bigcirc_{\ell=s'+1}^{s} \mathcal{Y}_{i_\ell}^t \Big) \Big(\CK_{\rR_{i_{s'}}^t \rQ_{i_{s'}}^{\neq t}}\Big( \mathcal{Y}_{i_{s'}}^t\bigotimes_{j\neq t}^r \mathcal{V}_{S}^j \Big)\big(\tilde{\rho}_{s'-1}\big)\CK_{\rR_{i_{s'}}^1 \rQ_{i_{s'}}^{\neq t}}\Big)\Big)\Big| = O(sn^{2k}\eps).
\end{align}
Finally, using Claim~\ref{claim:check} once more, we can remove the application of $\CK_{\rR_{i_{s'}}^1 \rQ_{i_{s'}}^{\neq t}}$ in the second term in~\eqref{eq:collapse-5} to obtain 
\begin{equation}\label{eq:collapse-6}
\Big|\Tr \Big( \CK_{i_{s}}\Big( \Big( \bigcirc_{\ell=s'+1}^{s} \mathcal{Y}_{i_\ell}^t \Big)\bigotimes_{j\neq t} \mathcal{Y}_{i_s}^j \Big) \big(\tilde{\rho}_{s'}\big)\Big) - \Tr\Big(\CK_{i_s} \Big( \Big( \bigcirc_{\ell=s'}^{s} \mathcal{Y}_{i_\ell}^t \Big) \bigotimes_{j\neq t} \mathcal{V}_{S}^j \Big)\big(\tilde{\rho}_{s'-1}\big)\Big)\Big| = O(sn^{2k}\eps).
\end{equation}
 To conclude, since the registers $\rQ_{i_{s'}}^{\neq t}$ are being traced out, and using Claim~\ref{claim:cidj} for the strategy $(U_i^j,V_{S}^j,\tilde{\rho}_{s'-1})$ the application of $\bigotimes_{j\neq t} \mathcal{V}_{S}^j$ in the second term in~\eqref{eq:collapse-6} can be replaced by $\bigotimes_{j\neq t} \mathcal{Y}_{i_s}^j$ for $\mathcal{Y}$ of our choice. Using the induction hypothesis~\eqref{eq:collapse-1} to bound the first term in~\eqref{eq:collapse-6}, this proves~\eqref{eq:collapse-1} for $s'$ and establishes the induction step, proving the claim. 
\end{proof}

The following corollary is a simple consequence of Claim~\ref{claim:cidj}. 

\begin{corollary}\label{cor:cidj-s}
Suppose the strategy $(U_i^j,V_{S}^j,\ket{\Psi})$ succeeds in test~\ref{step:test2} with probability at least $1-\eps$. Let $s$ be an integer and $i_1,\ldots,i_s\in [n]$. Then for any $t\in [r]$, set $S$ containing $i_s$, and choice of $\mathcal{Y}_{i_\ell}^t \in \{\mathcal{C}_{i_\ell}^t,\,\mathcal{D}_{i_\ell,S_{\ell}}^t|\,i_\ell\in S_\ell\}$ for $\ell\in[s-1]$,
\begin{equation}\label{eq:cidj-s}
\Big\|  \Big(\Big(  \mathcal{C}_{i_s}^t \bigcirc_{\ell=1}^{s-1} \mathcal{Y}_{i_\ell}^t \Big)\otimes \Id\Big)\big(\tilde{\rho}\big) - \Big(\Big(  \mathcal{D}_{i_s,S}^t \bigcirc_{\ell=1}^{s-1} \mathcal{Y}_{i_\ell}^t \Big)\otimes \Id\Big)\big(\tilde{\rho}\big)\Big\|_1 =  O\big(s^2n^{c_1+k}\eps\big),
\end{equation}
where $c_1$ is as in Claim~\ref{claim:collapse}. 
\end{corollary}

\begin{proof}
Using the freedom in the choice of the operators $\mathcal{Y}$,~\eqref{eq:collapse} from Claim~\ref{claim:collapse} shows that the strategy $(U_i^j,V_{S}^j,(\bigcirc_{\ell=1}^{s-1} \mathcal{Y}_{i_\ell}^t )\otimes \Id)(\tilde{\rho}))$ succeeds with probability $1- O\big(s^2 n^{c_1}\eps\big)$ in test~\ref{step:test2} of the protocol. The corollary then follows directly from Claim~\ref{claim:cidj}. 
\end{proof}

Our final claim shows that if the provers have a high success probability in both tests of protocol~$P$ the state $\sigma$ defined in~\eqref{eq:def-witness} must have low energy with respect to the local Hamiltonian $H$. 

\begin{claim}\label{claim:low-energy}
There exists a constant $c_2>0$ depending on $k$ only such that the following holds.
Let $\delta,\eps>0$ be such that the provers succeed in test~\ref{step:test1} of protocol~$P$ with probability at least $1-\delta$, and in test~\ref{step:test2} with probability at least $1-\eps$. Then 
$$\frac{1}{m}\,\Tr\big(H\sigma\big) \,=\, O\big(\delta + n^{c_2}  \eps\big).$$
\end{claim}

\begin{proof}
For any $i\in[n]$ we abbreviate $\DEC_{\rR_i^1\cdots\rR_i^r}$ as $\DEC_i$. By definition, 
$$\sigma = \Big(\bigotimes_{i=1}^n \DEC_{\rR_i^1\cdots\rR_i^r} \Big)\Big(\Tr_{\cup_t(\cup_i(\overline{\rR_i^t}\rQ_i^t) \rS^t) }\big(\tau\big)\Big),$$
where
$$\tau\,:=\,   \Big(\bigotimes_{j=1}^r \big(\cC_n^j\circ \cdots  \circ \cC_2^j \circ \cC_1^j \big)\Big)\big(\tilde{\rho}\big).$$
Fix a local term $H_j$ acting on $k$ qubits $S:=S_j=\{i_1,\ldots,i_k\}$, and let $U = [n]\backslash S$. Let $T_0=\emptyset$, and for $s=1,\ldots,k$ let $T_s = \{i_1,\ldots,i_s\}$. We show the following by induction on $s=0,\ldots,k$: 
\begin{equation}\label{eq:low-energy-0}
 \Big\| \Tr_{\rQ_U\overline{\rR}_U}\big(\tau\big) - \Tr_{\rQ_U\overline{\rR}_U}\Big(\bigotimes_{j=1}^r \big(\bigcirc_{\substack{i=1\\i\notin T_s}}^n \cC_i^j \big)\Big)\circ\Big(\bigotimes_{j=1}^r \mathcal{D}_{T_s,S}^j\Big)\big(\tilde{\rho}\big)\Big)\Big\|_1\,=\, O\big(s n^{c'_2}\eps\big),
\end{equation}
for some constant $c'_2>0$. The equation is trivially true for $s=0$. Suppose it true for some $s<k$. Let $\tilde{\rho}_s = \bigotimes_{j=1}^r \mathcal{D}_{T_s,S}^j(\tilde{\rho})$. We again proceed by induction on $\ell = i_{s+1},\ldots,1$ to show
\begin{align}
 \Big\|  \Tr_{\rQ_U\overline{\rR}_U}\Big(\Big(\bigotimes_{j=1}^r &\big(\bigcirc_{\substack{i=1\\i\notin T_s}}^{i_{s+1}} \cC_i^j \big)\Big)\big(\tilde{\rho}_s\big)\Big) -  \Tr_{\rQ_U\overline{\rR}_U}\Big(\Big(\bigotimes_{j=1}^r \big(\bigcirc_{\substack{i=\ell\\i\notin T_s}}^{i_{s+1}-1} \cC_i^j \big)\circ \mathcal{D}_{i_{s+1},S}^j  \bigcirc_{\substack{i=1\\i\notin T_s}}^{\ell-1} \cC_i^j \big)\Big)\big(\tilde{\rho}_s\big)\Big)\Big\|_1\notag\\
&= O\big((i_{s+1}+1-\ell) n^{c'_2-1}\eps\big).\label{eq:low-energy-1}
\end{align}
For $\ell = i_{s+1}$ the bound follows from Corollary~\ref{cor:cidj-s} applied to each of the provers, provided $c'_2$ is chosen large enough. For $\ell=1$, using $ \mathcal{D}_{i_{s+1},S}^j  \circ  \mathcal{D}_{T_{s},S}^j =  \mathcal{D}_{T_{s+1},S}^j $, together with the triangle inequality it establishes the induction step for the proof of~\eqref{eq:low-energy-0}. Suppose~\eqref{eq:low-energy-1} verified for some $\ell>1$. To prove it for $\ell-1$, we apply Corollary~\ref{cor:cidj-s} to the state 
$$\Big(\bigotimes_{j=1}^r \big(\bigcirc_{\substack{i=\ell\\i\notin T_s}}^{i_{s+1}-1} \cC_i^j \big)\circ \mathcal{D}_{i_{s+1},S}^j  \bigcirc_{\substack{i=1\\i\notin T_s}}^{\ell-1} \cC_i^j \big)\Big)\big(\tilde{\rho}_s\big)$$
 a number of times: first, the maps $\cC_{\ell-1}^j $ are replaced by $\cD_{\ell-1,S'}^j$ for some $S'$ containing both $\ell-1$ and $i_{s+1}$. Second, the $\mathcal{D}_{i_{s+1},S}^j$  are replaced by $\cD_{i_{s+1},S'}^j$. Next we use the relation $\cD_{i_{s+1},S'}^j\circ\cD_{\ell-1,S'}^j = \cD_{\ell-1,S'}^j\circ \cD_{i_{s+1},S'}^j$ and perform the same replacements in reverse. This establishes~\eqref{eq:low-energy-1} for $\ell-1$ (provided $c'_2$ is chosen large enough) and completes the induction step. We have now proven~\eqref{eq:low-energy-0}.

Using the definition of $\mathcal{D}_{T_k,S}$ and that both $H_j$ and $\DEC_{i_s}$, for $s=1,\ldots,k$, do not act on any of the registers in $\rQ_U$ or $\overline{\rR}_U$, from~\eqref{eq:low-energy-0} we get 
\begin{equation}\label{eq:low-energy-2}
\Big| \Tr\big(H_j\sigma\big) - \Tr\Big( H_j \Big(\bigotimes_{\ell=1}^k \DEC_{i_\ell}\Big) \Big( \bigotimes_{t=1}^r V_{S_j}^t\Big)\, \tilde{\rho}\, \Big(\bigotimes_{t=1}^r V_{S_j}^t\Big)^\dagger \Big)\Big| \,=\, O\big(n^{c'_2+1}\eps\big).
\end{equation}
By definition success $1-\delta$ in test~\ref{step:test1} of protocol~$P$ implies
$$
\frac{1}{m}\sum_{S_j = \{i_1,\ldots,i_k\}} \Tr\Big( H_j \Big(\bigotimes_{s=1}^k \DEC_{i_s} \Big) \Big( \Big(\bigotimes_{t=1}^r V_{S_j}^t \Big)\tilde{\rho}\Big(\bigotimes_{t=1}^r V_{S_j}^t \Big)^\dagger\Big)\Big) \leq \delta,
$$
which together with~\eqref{eq:low-energy-2} proves the claim for an appropriate choice of $c_2$.  
\end{proof}

Lemma~\ref{lem:soundness} now follows directly from Claim~\ref{claim:low-energy} and the fact that any strategy with success $1-\eps$ in Protocol~$P$ must have success probability at least $1-2\eps$ in each of the two tests~\ref{step:test1} and~\ref{step:test2} of the protocol.

\bibliographystyle{alpha}
\bibliography{qma_mip}
\end{document}